\pdfoutput=1

\documentclass [smallcondensed,natbib,final]{svjour3}

\journalname{Springer}

\usepackage[colorlinks,linkcolor=blue,citecolor=blue,urlcolor=blue]{hyperref}
\usepackage{amsmath,amssymb}
\usepackage{mathptmx}
\usepackage{graphicx}
\usepackage{algorithm,algorithmic}

\smartqed

\newcommand{\eps}{\ensuremath{\epsilon}}
\newcommand{\abs}[1]{\ensuremath{\lvert{#1}\rvert}}

\newcommand{\EE}{\ensuremath{\mathbb{E}}}
\newcommand{\PP}{\ensuremath{\mathbb{P}}}

\usepackage{anyfontsize}
\newcommand{\one}{\mbox{1\hspace{-0.38em}\fontsize{10.5}{10}\selectfont\textrm{1}}}

\begin{document}

\title{Efficient generation of random de\-range\-ments with the expected distribution of cycle lengths}

\titlerunning{Efficient generation of random de\-range\-ments}

\author{J. Ricardo G. Mendon\c{c}a\,\href{https://orcid.org/0000-0002-5516-0568}{\includegraphics[scale=0.3]{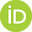}}}

\authorrunning{J. R. G. Mendon\c{c}a}

\institute{J. Ricardo G. Mendon\c{c}a \at Escola de Artes, Ci\^{e}ncias e Humanidades, Universidade de S\~{a}o Paulo, 03828-000 S\~{a}o Paulo, SP, Brazil \\ \email{jricardo@usp.br}}

\date{Submitted: \today}

\maketitle

\begin{abstract}
We show how to generate random de\-range\-ments efficiently by two different techniques: random restricted transpositions and sequential importance sampling. The algorithm employing restricted transpositions can also be used to generate random fixed-point-free involutions only, a.\,k.\,a.\ random perfect matchings on the complete graph. Our data indicate that the algorithms generate random samples with the expected distribution of cycle lengths, which we derive, and for relatively small samples, which can actually be very large in absolute numbers, we argue that they generate samples indistinguishable from the uniform distribution. Both algorithms are simple to understand and implement and possess a performance comparable to or better than those of currently known methods. Simulations suggest that the mixing time of the algorithm based on random restricted transpositions (in the total variance distance with respect to the distribution of cycle lengths) is $O(n^{a}\log{n}^{2})$ with $a \simeq \frac{1}{2}$ and $n$ the length of the de\-range\-ment. We prove that the sequential importance sampling algorithm generates random de\-range\-ments in $O(n)$ time with probability $O(1/n)$ of failing.

\keywords{Restricted permutation \and random transposition walk \and random perfect matching \and switch Markov chain \and mixing time}

\subclass{65C60 \and 68R05 \and 05A05}

\end{abstract}


\section{\label{intro}Introduction}

De\-range\-ments are permutations $\sigma = \sigma_{1} \cdots \sigma_{n}$ on integer $n \geq 2$ labels such that $\sigma_{i} \ne i$ for all $i = 1, \dots, n$. De\-range\-ments are useful in a number of applications like in the testing of software branch instructions and random paths and data randomization and experimental design \citep{bacher,pd-rlg-sph,sedgewick}. A well known algorithm to generate random derange\-ments is Sattolo's algorithm, that outputs a random cyclic de\-range\-ment in $O(n)$ time \citep{gries,prodinger,sattolo,wilson}. An $O(2n)$ algorithm to generate random de\-range\-ments in general (not only cyclic de\-range\-ments) has been given by \citet{analco} and \citet{iran}. Algorithms to generate all $n$-de\-range\-ments in lexicographic or Gray order have also been developed \citep{akl,baril,korsh}.

In this paper we propose two procedures to generate random de\-range\-ments with the expected distribution of cycle lengths: one based on the randomization of de\-range\-ments by random restricted transpositions (a random walk in the set of de\-range\-ments) and the other based on a simple sequential importance sampling scheme. The generation of restricted permutations by means of sequential importance sampling is closely related with the problem of estimating the permanent of a $0$-$1$matrix, an important problem in, e.\,g., graph theory, statistical mechanics, and experimental design \citep{beichl,brualdi,pd-rlg-sph}. Simulations show that the randomization algorithm samples a de\-range\-ment in $O(n^{a}\log{n}^{2})$ time, where $n$ is the size of the de\-range\-ment and $a \simeq \frac{1}{2}$, while the sequential importance sampling algorithm does it in $O(n)$ time but with a small probability $O(1/n)$ of failing. The algorithms are straighforward to understand and implement and can be modified to perform related computations of interest in many areas.

Throughout the paper we employ the expected distribution of cycle lengths to analyse the algorithms because they are such fundamental invariants of permutations from which many other statistics can be derived, for instance, the expected number of ascents, descents, or transpositions, because they offer a sufficiently ``aggregate,'' not too detailed quantity, and also because we have simple exact expressions for the probability of observing derangements with a given number of cycles with which we can compare the numerical data.


\section{\label{sec:math}Mathematical preliminaries}

Let us briefly recapitulate some notation and terminology on permutations. Detailed accounts suited to our needs are given by \citet{arratia} and \citet{charalambides}.

We denote a permutation of a set of integer $n \geq 2$ labels (an $n$-permutation), formally a bijection of $[\,n\,] = \{1, \dots, n\}$ onto itself, by $\sigma = \sigma_{1} \cdots \sigma_{n}$, where $\sigma_{i}=\sigma(i)$. If $\sigma$ and $\pi$ are two $n$-permutations, their product is given by the composition $\sigma\pi = \sigma(\pi_{1}) \cdots \sigma(\pi_{n})$. A cycle of length $k \leq n$ in a $n$-permutation $\sigma$ is a sequence of indices $i_{1}, \dots, i_{k}$ such that $\sigma(i_{1})=i_{2}$, \dots, $\sigma(i_{k-1})=i_{k}$, and $\sigma(i_{k})=i_{1}$, completing the cycle. Fixed points are $1$-cycles, transpositions are $2$-cycles. An $n$-permu\-ta\-tion with $a_{k}$ cycles of length $k$, $ 1 \leq k \leq n$, is said to be of type $(a_{1}, \dots, a_{n})$, with $\sum_{k}ka_{k}=n$. For example, the $9$-permutation $174326985=(1)(43)(6)(8)(9527)$ has $5$ cycles and is of type $(3,1,0,1)$, where we have omitted the trailing $a_{5} = \cdots = a_{9}=0$. Note that in the cycle notation $(1)(43)(6)(8)(9527)$ the parentheses are unnecessary, since each new left-to-right absolute maximum $\sigma_{i} > \max\{\sigma_{1}, \dots, \sigma_{i-1}\}$ corresponds to a new cycle (the so-called Foata's representation).

The number of $n$-permutations with $k$ cycles is given by the unsigned Stirling number of the first kind $n \brack k$. Useful formulae involving these numbers are ${0 \brack 0}=1$, ${n \brack 0}=0$, and the recursion relation ${n+1 \brack k} = n{n \brack k} + {n \brack k-1}$. We have ${n \brack n}=1$, counting just the identity permutation $\text{id}=(1)(2)\cdots(n)$, ${n \brack n-1}={n \choose 2}$, counting $n$-permuta\-tions with $n-2$ fixed points, that can be taken in ${n \choose n-2}={n \choose 2}$ different ways, plus a transposition of the remaining two labels, and ${n \brack 1}=(n-1)!$, the number of cyclic $n$-de\-range\-ments. It can also be shown that ${n \brack 2}=(n-1)!H_{n-1}$, where $H_{k} = 1+\frac{1}{2}+\cdots+\frac{1}{k}$ is the $k\,$th harmonic number. Obviously, ${n \brack 1} + \cdots + {n \brack n} = {n+1 \brack 1} = n!$, the total number of $n$-permutations.

Let us denote the set of all $n$-de\-range\-ments by ${D}_{n}$. It is well known that
\begin{equation}
\label{eq:rencontre}
{{d}_{n}} = \abs{D_{n}} = n!\Big(1-\frac{1}{1!} + \cdots + \frac{(-1)^{n}}{n!}\Big) = \Big\lfloor\frac{n!+1}{e}\Big\rfloor, \quad n \geq 1,
\end{equation}
the \textit{rencontres} numbers, where the floor function $\lfloor x \rfloor$ evaluates to the greatest integer less than or equal to $x$. Let us also denote the set of $k$-cycle $n$-derange\-ments, irrespective of their type, by ${D}_{n}^{(k)}$. The ${D}_{n}^{(k)}$ are disjoint with ${D}_{n}^{(k)}=\varnothing$ for $k > \lfloor n/2 \rfloor$. If we want to generate random $n$-de\-range\-ments over ${D}_{n} = {D}_{n}^{(1)} \cup \cdots \cup {D}_{n}^{(\lfloor n/2\rfloor)}$, we must be able to generate $k$-cycle random $n$-de\-range\-ments with probabilities
\begin{equation}
\label{eq:prob}
\PP(\sigma \in {D}_{n}^{(k)}) = \frac{d_{n}^{(k)}}{d_{n}}, 
\end{equation}
where $d_{n}^{(k)}=\abs{{D}_{n}^{(k)}}$. The following proposition establishes the cardinality of the sets ${D}_{n}^{(k)}$.

\begin{proposition}
\label{prop:dnk}
The cardinality of the set ${D}_{n}^{(k)}$ is given by
\begin{equation}
\label{eq:dnk}
d_{n}^{(k)} = \sum_{j=0}^{k}(-1)^{j}{n \choose j}{n-j \brack k-j}.
\end{equation}
\end{proposition}
\begin{proof}
The number of $n$-permutations with $k$ cycles is $n \brack k$. Of these, $n{n-1 \brack k-1}$ have at least one fixed point, ${n \choose 2}{n-2 \brack k-2}$ have at least two fixed points, and so on. Perusal of the inclusion-exclusion principle furnishes the result. \qed
\end{proof}

\begin{proposition}
\label{prop:rec}
The numbers ${d}_{n}^{(k)}$ obey the recursion relation
\begin{equation}
\label{eq:rec}
d_{n+1}^{(k)} = n\big(d_{n}^{(k)} + d_{n-1}^{(k-1)}\big)
\end{equation}
with $d_{0}^{(0)}=1$ and $d_{n}^{(0)}=0$, $n \geq 1$.
\end{proposition}
\begin{proof}
We give a combinatorial argument. Consider the de\-range\-ment of $n+1$ labels with $k$ cycles enumerated by $d_{n+1}^{(k)}$ according to the condition of the largest label $n+1$. This largest label is either in a $2$-cycle or it is not. If it is, it is attaching a $2$-cycle to an $n-1$-de\-range\-ment with $k-1$ cycles, of which there are $d_{n-1}^{(k-1)}$, and since it can be paired up with any of other $n$ possible labels, it contributes a factor $n\,d_{n-1}^{(k-1)}$ to $d_{n+1}^{(k)}$. If label $n+1$ is not in a $2$-cycle, then it belongs to one of the $k$ cycles of an otherwise $n$-de\-range\-ment, of which there are $d_{n}^{(k)}$, and since in each cycle of length $r$ there are $r$ possible places to insert label $n+1$ (the first and the last places within each cycle coincide) and $\sum{r}=n$, it contributes a factor $n\,d_{n}^{(k)}$ to $d_{n+1}^{(k)}$. Adding the two contributions furnishes the recursion relation (\ref{eq:rec}).
\qed
\end{proof}

The numbers ${d}_{n}^{(k)}$ are sometimes called associated Stirling number of the first kind. Equation~(\ref{eq:dnk}) recovers ${{d}_{n}^{(0)}}=0$ and ${{d}_{n}^{(1)}}={n \brack 1}=(n-1)!$ for $n \geq 1$, while we find that ${{d}_{n}^{(2)}}=(n-1)!(H_{n-2}-1)$ for $n \geq 2$. Equation~(\ref{eq:rec}) generalizes the recursion relation $d_{n+1} = {n(d_{n}+d_{n-1})}$ for the \textit{rencontres} numbers. A notable identity, valid for $n$ even, is ${{d}_{n}^{(n/2)}} = {(n-1)(n-3) \cdots 3 \cdot 1}$, the number of fixed-point-free involutions $\sigma$ such that $\sigma^{2} = \text{id}$, a.\,k.\,a. perfect matchings, see Section~\ref{sec:match}. From Eqs.~(\ref{eq:rencontre})--(\ref{eq:dnk}) we see that already for small $n$ we obtain $\PP(\sigma \in {D}_{n}^{(1)}) \simeq e/n$ and $\PP(\sigma \in {D}_{n}^{(2)}) \simeq (H_{n-2}-1)e/n$.
\begin{remark}
\label{rmk:cauchy}
One could consider the distribution of $n$-de\-range\-ments over possible cycle types (instead of cycle lengths) for a ``finer'' view of the distribution. The number of $n$-permutations of type $(a_{1}, \dots, a_{n})$ is given by Cauchy's formula
\begin{equation}
\label{eq:cycles}
k_{n}(a_{1}, \dots, a_{n}) = \frac{n!}{1^{a_{1}}a_{1}!\, \cdots\, n^{a_{n}}a_{n}!}.
\end{equation}
The analogue of (\ref{eq:prob}) is given by $\PP(\sigma \in {K}_{n}(0, a_{2}, \dots, a_{n})) = k_{n}(0, a_{2}, \dots, a_{n})/{{d}_{n}}$, where \linebreak ${K}_{n}(0, a_{2}, \dots, a_{n})$ is the conjugacy class formed by all $n$-permutations of type $(0, a_{2}, \dots, a_{n})$.


Other permutation statistics, e.\,g.\ the expected number of descents or transpositions (every cycle of length $k$ factors into $k-1$ transpositions, though) could be employed as well (see, for instance, \citet[Sec.~1.5]{matchings} for a connection with integer partitions) but they would lead to more complicate expressions for $\PP(\sigma \in \,\cdot\,)$.
\end{remark}


\section{\label{sec:random}Generating random de\-range\-ments by random transpositions}

\subsection{\label{sec:walk}The random transposition walk}

Our first approach to generate random $n$-deran\-ge\-ments consists in taking an initial $n$-de\-range\-ment and to scramble it by random restricted transpositions enough to obtain a sample distributed over $D_{n}$ according to the probabilities given by (\ref{eq:prob}). By restricted transpositions we mean swaps $\sigma_{i} \leftrightarrow \sigma_{j}$ avoiding pairs for which $\sigma_{i}=j$ or $\sigma_{j}=i$. Algorithm~\ref{alg:switch} describes the generation of random $n$-de\-range\-ments according to this idea, where $\mathit{mix}$ is a constant establishing the amount of random restricted transpositions to be attempted and $\mathit{rnd}$ is a computer generated pseudorandom uniform deviate in $(0,1)$.
\begin{remark}
\label{rmk:three}
Algorithm~\ref{alg:switch} is applicable only for $n \geq 4$, since it is not possible to connect the even permutations $231$ and $312$ by a single transposition.
\end{remark}

A good choice for the initial de\-range\-ment in Algorithm~\ref{alg:switch} is any cyclic de\-range\-ment (cycle length $k=1$), for example, $\sigma = (2\:3 \cdots n\,1)$. A particularly bad choice would be an involution ($n$ even, all cycle lengths $k=2$), for example, $\sigma = (n~n-1) \cdots (2~1)$, because then the algorithm would not be able to generate de\-range\-ments with $k \ne 2$. Incidentally, this suggests the use of Algorithm~\ref{alg:switch} to generate random fixed-point-free involutions, a.\,k.\,a.\ random perfect matchings, see Section~\ref{sec:match}. To avoid this problem we hardcoded the requirement that Algorithm~\ref{alg:switch} starts with a cyclic de\-range\-ment. If several parallel streams of random de\-range\-ments are sought, one can set different initial random cyclic de\-range\-ments from a one-line implementation of Sattolo's algorithm.
\begin{remark}
\label{rmk:split}
The minimum number of restricted transpositions necessary to take a cyclic $n$-deran\-ge\-ment into a $k$-cycle $n$-derange\-ment is $k-1$, $1 \leq k \leq \lfloor n/2 \rfloor$, since transpositions of labels that belong to the same cycle split it into two cycles,
\begin{equation}
(ab)(i_{1} \cdots i_{a-1}i_{a}i_{a+1} \cdots i_{b-1}i_{b}i_{b+1} \cdots i_{k}) =(i_{1} \cdots i_{a-1}i_{b}i_{b+1} \cdots i_{k})(i_{a+1} \cdots i_{b-1}i_{a})
\end{equation}
and, conversely, transpositions involving labels of different cycles join them into a single cycle. If Algorithm~\ref{alg:switch} is started with a cyclic de\-range\-ment then one must set $\mathit{mix} \geq n/2$.
\end{remark}

\renewcommand{\thealgorithm}{T}
\begin{algorithm}[t]
\caption{Random de\-range\-ments by random restricted transpositions}
\label{alg:switch}
\algsetup{indent=1.5em,linenosize=\small}
\begin{algorithmic}[1]
\REQUIRE Initial cyclic $n$-de\-range\-ment $\sigma_{1} \sigma_{2} \cdots \sigma_{n}$ \label{alg:switch:req}
\STATE $\mathit{mix} \gets$ number of restricted transpositions to attempt
\FOR{$m=1$ \TO $\mathit{mix}$} \label{alg:switch:mix}
   \STATE $i \gets \lceil \mathit{rnd} \cdot n\rceil$,
          $j \gets \lceil \mathit{rnd} \cdot n\rceil$
   \IF {$(\sigma_{i} \ne j) \land (\sigma_{j} \ne i)$}
      \STATE swap $\sigma_{i} \leftrightarrow \sigma_{j}$
   \ENDIF
\ENDFOR
\ENSURE For sufficiently large $\mathit{mix}$, $\sigma_{1}\cdots\sigma_{n}$ is a ``sufficiently random'' de\-range\-ment from $D_{n}$
\end{algorithmic}
\end{algorithm}

Algorithm~\ref{alg:switch} ensures that for a sufficiently large constant $\mathit{mix}$ it generates a ``sufficiently random'' de\-range\-ment from $D_{n}$. Slurring over the philosophical questions related with the concept of randomness, in this paper ``sufficiently random'' means with cycle lengths distributed according to the exact probabilities given by Eq.~(\ref{eq:prob}). We make this statement about Algorithm~\ref{alg:switch} more precise in Section~\ref{sec:mix}. Clearly, the correct distribution of cycle lengths is a necessary but not sufficient condition for uniformity ($\PP(\sigma \in D_{n}) = 1/d_{n}$) over $D_{n}$, and we do not claim uniformity for Algorithm~\ref{alg:switch} here or elsewhere in this paper. This point is elaborated further in Remark~\ref{rmk:nlnn} and Section~\ref{sec:uniform}.

We run Algorithm~\ref{alg:switch} for $n=64$ and different values of $\mathit{mix} \geq n$ and collect data. Simulations were performed on Intel Xeon E5-1650~v3 processors running \texttt{-O3} compiler-optimized C code (GCC v.~$7.4.0$) over Linux kernel $4.15.14$ at $3.50$\,GHz, while the numbers (\ref{eq:dnk}) were calculated on the software package Mathematica~11.3 \citep{wolfram}. We draw our pseudorandom numbers from Vigna's superb \texttt{xoshiro256+} generator \citep{xoshiro}. Our results appear in Table~\ref{tab:rnd}. We see from that table that with $\mathit{mix}=n$ random restricted transpositions there is a slight excess of probability mass in the lower $k$-cycle sets with $k=1, 2$, and $3$. Trying to scramble the initial $n$-derrangement by $2n$ restricted transpositions performs better. The difference between attempting $2n$ and ${n}\log{n}$ random restricted transpositions is much less pronounced. Figures for de\-range\-ments of higher cycle number fluctuate more due to the finite size of the sample. The data suggest that Algorithm~\ref{alg:switch} can generate a random $n$-de\-range\-ment uniformly distributed on ${D}_{n}$ with $2n$ random restricted transpositions, employing $4n$ pseudorandom numbers in the process. This is further discussed in Section~\ref{sec:mix}.
\begin{remark}
\label{rmk:nlnn}
It is a classic result that $O({n}\log{n})$ transpositions are needed before an unrestricted shuffle by transpositions becomes ``sufficiently random'' \citep{aldous,shahshahani}. A similar analysis for random transpositions over de\-range\-ments is complicated by the fact that de\-range\-ments do not form a group. Recently, the analysis of the spectral gap of the Markov transition kernel of the process provided the upper bound $\mathit{mix} < Cn+a{n}\log{n^2}$, with $a>0$ and $C \geq 0$ a decreasing function of $n$ \citep{aaron}. This bound results from involved estimations and approximations and may not be very accurate. Related results for the mixing time of the random transposition walk over permutations with one-sided restrictions $\sigma_{i} \geq b_{i}$ for given $n \geq b_{n} \geq \cdots \geq b_{1} \geq 1$---a pattern known as a Ferrer's board in the combinatorics literature---appear in \citep{olena,hanlon}. Recently the case $i-t \leq \sigma_{i} \leq i+1$, $t \geq 1$, has been treated by \citet{chung}, although they do not explore the mixing times of the associated Markov chains.
\end{remark}

\begin{table}[t]
\caption{\label{tab:rnd}Proportion of $n$-de\-range\-ments in ${D}_{n}^{(k)}$ measured in $10^{10}$ samples generated by Algorithms~\ref{alg:switch} and \ref{alg:sis} for $n=64$. The notation $x_{-a}$ reads $x \times 10^{-a}$. Data for Algorithm~\ref{alg:sis} are based on a run with a ratio of completed/attempted de\-range\-ments of $0.985472$.} 
\centering
\begin{tabular}{ccccccccc}
\hline
Cycles & & \multicolumn{3}{c}{Algorithm~\ref{alg:switch} ($\mathit{mix}$)} & & \multicolumn{1}{c}{Algorithm~\ref{alg:sis}} & & \multicolumn{1}{c}{Exact} \\
\cline{1-1} \cline{3-5} \cline{7-7} \cline{9-9}
$k$ & & $n$ & $2n$ & ${n}\log{n}$ & & --- & & Eqs.~(\ref{eq:rencontre})--(\ref{eq:dnk}) \\
\cline{1-1} \cline{3-5} \cline{7-7} \cline{9-9}
 $1$ & {} & $0.042\,933$ & $0.042\,479$ & $0.042\,473$ & {} & $0.042\,475$ & {} & $0.042\,473$ \\
 $2$ & {} & $0.158\,395$ & $0.157\,691$ & $0.157\,679$ & {} & $0.157\,684$ & {} & $0.157\,677$ \\
 $3$ & {} & $0.260\,129$ & $0.258\,787$ & $0.258\,765$ & {} & $0.258\,788$ & {} & $0.258\,772$ \\
 $4$ & {} & $0.252\,739$ & $0.253\,304$ & $0.253\,305$ & {} & $0.253\,306$ & {} & $0.253\,301$ \\
 $5$ & {} & $0.167\,189$ & $0.167\,621$ & $0.167\,639$ & {} & $0.167\,622$ & {} & $0.167\,635$ \\
 $6$ & {} & $0.079\,498$ & $0.080\,390$ & $0.080\,402$ & {} & $0.080\,389$ & {} & $0.080\,400$ \\
 $7$ & {} & $0.028\,825$ & $0.029\,192$ & $0.029\,195$ & {} & $0.029\,196$ & {} & $0.029\,200$ \\
 $8$ & {} & $0.008\,087$ & $0.008\,269$ & $0.008\,274$ & {} & $0.008\,272$ & {} & $0.008\,274$ \\
 $9$ & {} & $0.001\,821$ & $0.001\,868$ & $0.001\,869$ & {} & $0.001\,868$ & {} & $0.001\,869$ \\
$10$ & {} & $3.292_{-4}$ & $3.416_{-4}$ & $3.418_{-4}$ & {} & $3.412_{-4}$ & {} & $3.417_{-4}$ \\
$11$ & {} & $4.914_{-5}$ & $5.109_{-5}$ & $5.120_{-5}$ & {} & $5.103_{-5}$ & {} & $5.116_{-5}$ \\
$12$ & {} & $5.997_{-6}$ & $6.322_{-6}$ & $6.301_{-6}$ & {} & $6.354_{-6}$ & {} & $6.326_{-6}$ \\
$13$ & {} & $6.215_{-7}$ & $6.493_{-7}$ & $6.301_{-7}$ & {} & $6.507_{-7}$ & {} & $6.499_{-7}$ \\
$14$ & {} & $4.83_{-8}$ & $5.40_{-8}$ & $5.57_{-8}$ & {} & $5.44_{-8}$ & {} & $5.569_{-8}$ \\
$15$ & {} & $4.6_{-9}$ & $3.1_{-9}$ & $3.0_{-9}$ & {} & $4.1_{-9}$ & {} & $3.989_{-9}$ \\
$16$ & {} & $4_{-10}$ & $1_{-10}$ & $3_{-10}$ & {} & $1_{-10}$ & {} & $2.390_{-10}$ \\
\hline
\end{tabular}
\end{table}


\subsection{\label{sec:match}The perfect matching connection}

In Sec~\ref{sec:walk} we remarked that if one seeds Algorithm~\ref{alg:switch} with an initial fixed-point-free involution, i.\,e., a de\-range\-ment with all cycle lengths equal to $2$, then all subsequent de\-range\-ments generated by the algorithm will also be fixed-point-free involutions. Such de\-range\-ments are in $1$--$1$ correspondence with perfect matchings on a complete graph, since any unoriented edge $\sigma_{i}\sigma_{j}$ can occur. A perfect matching on a graph is a set of disjoint edges of the graph containing all its vertices. The connection between permutations with retricted positions and perfect matchings is well known \citep{brualdi,lovasz} and has been explored recently in the context of random walks on trees and applications, including Monte Carlo estimation of hard enumeration problems \citep{chung,phylos,matchings,kolesnik,dyer,muller}.

Cauchy's formula (\ref{eq:cycles}) gives the number of perfect matchings on a complete graph of even number $n$ of vertices as the number of derangements with $n/2$ cycles of length $2$,
\begin{equation}
\label{eq:match}
k_{n}(0,n/2,0,\dots,0) = \frac{n!}{2^{n/2}(n/2)!} \sim \sqrt{2(n/e)^{n}},
\end{equation}
where the asymptotics follows from Stirling's approximation $n! \simeq \sqrt{2\pi n}\,(n/e)^{n}$. The number (\ref{eq:match}) can also be understood as the number of partitions of a set of even size $n$ into $n/2$ unordered parts of size $2$ each---which is just another definition of a perfect matching. We see that the probability that a random derangement is a perfect matching is very small,
\begin{equation}
\label{eq:pmatch}
\PP(\sigma \in D_{n}^{(n/2)}) = \frac{k(0,n/2,0,\dots,0)}{d_{n}} \simeq \frac{e}{\sqrt{\pi n}}\sqrt{(e/n)^{n}}.
\end{equation}
For example, for $n=10$ equation (\ref{eq:pmatch}) gives a $1$ in $1389$ chance that a random derangement is a perfect matching. If one employs a standard algorithm to generate random permutations, the chance that it outputs a random perfect matching decreases to $1$ in $3777$. With a simple tweak, though, Algorithm~\ref{alg:switch} can generate random perfect matchings on the complete graph at will. Although this is not a particularly difficult computational problem, having a simple and efficient algorithm to generate such random perfect matchings might be useful.


\section{\label{sec:sis}Sequential importance sampling of de\-range\-ments}

\subsection{\label{sec:alg}The SIS algorithm}

Sequential importance sampling (SIS) is an importance sampling scheme with the sampling weights built up sequentially. The idea is particularly suited to sample composite objects $X=X_{1} \cdots X_{n}$ from a complicated sample space $\cal{X}$ for which the high-dimensional volume $\abs{\cal{X}}$, from which the uniform distribution $\PP(X) = \abs{\cal{X}}^{-1}$ follows, may not be easily calculable. However, since we can always write
\begin{equation}
\label{eq:ppp}
\PP(X_{1} \cdots X_{n}) = \PP(X_{1}) \PP(X_{2} \mid X_{1}) \cdots \PP(X_{n} \mid X_{1} \cdots X_{n-1}),
\end{equation}
we can think of ``telescoping'' the sampling of $X$ by first sampling $X_{1}$, then use the updated information brought by the knowledge of $X_{1}$ to sample $X_{2}$ and so on. In Monte Carlo simulations, the right-hand side of (\ref{eq:ppp}) actually becomes $\PP_{1}(X_{1}) \PP_{2}(X_{2} \mid X_{1}) \cdots$ $\PP_{n}(X_{n} \mid X_{1} \cdots X_{n-1})$, with the distributions ${\PP}_{i}(\,\cdot\,)$ estimated or inferred incrementally based on approximate weighting functions for the partial objects $X_{1} \cdots X_{i-1}$. Expositions of the SIS framework of interest to what follows appear in \citet{cdhl,pd-rlg-sph}.

Algorithm~\ref{alg:sis} describes a SIS algorithm to generate random de\-range\-ments inspired by the analogous problem of sampling contingency tables with restrictions \citep{cdhl,pd-rlg-sph} as well as by the problem of estimating the permanent of a matrix \citep{beichl,cdhl,kuznetsov,rasmussen}. Our presentation of Algorithm~\ref{alg:sis} is not the most efficient for implementation; the auxiliary sets $J_{i}$, for instance, are not actually needed and were included only to facilitate the analysis of the algorithm, and the $n$ tests in line~\ref{alg:sis:if} can be reduced to a single test in the last pass, since all $J_{i} \ne \varnothing$ except perhaps $J_{n}$.

The distribution of cycle lengths in $10^{10}$ de\-range\-ments generated by Algorithm~\ref{alg:sis} is presented in Table~\ref{tab:rnd}. We see excellent agreement between the data and the expected values.

\renewcommand{\thealgorithm}{S}
\begin{algorithm}[t]
\caption{Random de\-range\-ments by sequential importance sampling}
\label{alg:sis}
\algsetup{indent=1.5em,linenosize=\small}
\begin{algorithmic}[1]
\STATE $J \gets [\,n\,]$
\FOR{$i=1$ \TO $n$} \label{alg:sis:for}
   \STATE {$J_{i} \gets J \setminus \{i\}$} \label{alg:sis:ji}
   \IF {$J_{i} \ne \varnothing$} \label{alg:sis:if}
      \STATE choose $j_{i} \in J_{i}$ uniformly at random \label{alg:sis:uar}
      \STATE $\sigma_{i} \gets j_{i}$ \label{alg:sis:sig}
      \STATE $J \gets J \setminus \{j_{i}\}$ \label{alg:sis:update}
   \ELSE
      \STATE fail
   \ENDIF
\ENDFOR \label{alg:sis:endfor}
\ENSURE If completed, $\sigma_{1}\cdots\sigma_{n}$ is a ``sufficiently random'' de\-range\-ment from $D_{n}$
\end{algorithmic}
\end{algorithm}

\subsection{\label{sec:fail}Failure probability of the SIS algorithm}

In the $i\,$th pass of the loop in Algorithm~\ref{alg:sis}, $\sigma_{i}$ can pick (lines~\ref{alg:sis:uar}--\ref{alg:sis:sig}) one of either $n-i$ or $n-i+1$ labels, depending on whether label $i$ has already been picked. This guarantees the construction of the $n$-de\-range\-ment up to the $(n-1)$st label $\sigma_{n-1}$. The $n$-deran\-ge\-ment is completed only if the last remaining label is different from $n$, such that $\sigma_{n}$ does not pick $n$. The probability that Algorithm~\ref{alg:sis} fails is thus given by
\begin{equation}
\label{eq:delta}
\PP(\sigma_{n}=n \mid \sigma_{1} \cdots \sigma_{n-1}) = \PP(\sigma_{1} \ne n)\, \PP(\sigma_{2} \ne n \mid \sigma_{1}) \cdots \PP(\sigma_{n-1} \ne n \mid \sigma_{1} \cdots \sigma_{n-2}).
\end{equation}
Now, according to Algorithm~\ref{alg:sis}, line~\ref{alg:sis:uar}, we have
\begin{equation}
\label{eq:ps}
\PP(\sigma_{i} \ne n \mid \sigma_{1} \cdots \sigma_{i-1}) =
1-\PP(\sigma_{i}=n \mid \sigma_{1} \cdots \sigma_{i-1}) =
1-\frac{1}{\EE(\abs{J_{i}(\sigma_{1} \cdots \sigma_{i-1})})},
\end{equation}
where $\EE(\abs{J_{i}(\sigma_{1} \cdots \sigma_{i-1})})$ is the expected size of the set $J_{i}$ in the $i\,$th pass of the loop in Algorithm~\ref{alg:sis}. The failure probability then becomes
\begin{equation}
\label{eq:fail}
\PP(\sigma_{n}=n \mid \sigma_{1} \cdots \sigma_{n-1}) = \prod_{i=1}^{n-1}\Big(1-\frac{1}{E_{i}}\Big),
\end{equation}
where $E_{i}$ stands for $\EE(\abs{J_{i}(\sigma_{1} \cdots \sigma_{i-1})})$.


The computation of (\ref{eq:fail}) is a cumbersome business and we will not pursued it here. The following theorem establishes an upper bound on the failure probability of Algorithm~\ref{alg:sis}.
\begin{theorem}
Algorithm~\ref{alg:sis} fails with probability $O(1/n)$.
\end{theorem}
\begin{proof}
In the $i\,$th pass of the loop in Algorithm~\ref{alg:sis} we have
\begin{equation}
\label{eq:indic}
\abs{J_{i}(\sigma_{1} \cdots \sigma_{i-1})} = n-i+\sum_{j=1}^{i-1}\one\{(\sigma_{j}=i \mid \sigma_{1}, \dots, \sigma_{j-1})\},
\end{equation}
where the symbol $\one\{A\}$ stands for the indicator function that equals $1$ if $A$ occurs and $0$ if $A$ does not occur. We thus have that $\abs{J_{i}(\sigma_{1} \cdots \sigma_{i-1})} = n-i$ or $n-i+1$, such that the expectation $E_{i} = \EE(\abs{J_{i}(\sigma_{1} \cdots \sigma_{i-1})})$ obeys
\begin{equation}
1-\frac{1}{n-i} < 1-\frac{1}{E_{i}} < 1-\frac{1}{n-i+1}
\end{equation}
and it immediately follows that
\begin{equation}
\label{eq:bound}
\PP(\sigma_{n}=n \mid \sigma_{1} \cdots \sigma_{n-1}) = \prod_{i=1}^{n-1}\Big(1-\frac{1}{E_{i}}\Big) < \prod_{i=1}^{n-1}\Big(1-\frac{1}{n-i+1}\Big) = \frac{1}{n}.
\end{equation} \qed
\end{proof}

We can obtain a slightly better bound for $\PP(\sigma_{n}=n \mid \sigma_{1} \cdots \sigma_{n-1})$. The difficulty in the calculation of $\EE(\abs{J_{i}(\sigma_{1} \cdots \sigma_{i-1})})$ resides in the calculation of $\EE(\one\{(\sigma_{j}=i \mid \sigma_{1}, \dots, \sigma_{j-1})\})$. We can approximate this calculation by ignoring the conditioning of the event $(\sigma_{j}=i)$ on the event $(\sigma_{1} \cdots \sigma_{j-1})$, i.\,e., by ignoring correlations between the values assumed by the $\sigma_{j}$ along a ``path'' in the algorithm. The approximation is clearly better in the beginning of the construction of $\sigma$, when $j$ is small, than later. We get
\begin{equation}
\label{eq:indep}
\begin{split}
\EE(\abs{J_{i}(\sigma_{1} \cdots \sigma_{i-1})}) &= n-i + \sum_{j=1}^{i-1}\EE\big(\one\{(\sigma_{j}=i \mid \sigma_{1}, \dots, \sigma_{j-1})\}\big) \\ & \approx n-i+\sum_{j=1}^{i-1}\EE\big(\one\{\sigma_{j}=i\}\big) = n-i+\frac{i-1}{n-1}.
\end{split}
\end{equation}
This approximate $E_{i}$ is greater than the true $E_{i}$, because conditioning $J_{i}$ on $(\sigma_{1} \cdots \sigma_{i-1})$ can only restrict the set of indices available to $\sigma_{i}$, not enlarge it. The approximate value of $1-1/E_{i}$ is thus greater than its true value, and we can bound the failure probability (\ref{eq:fail}) by
\begin{equation}
\label{eq:approx}
\PP(\sigma_{n} = n \mid \sigma_{1} \cdots \sigma_{n-1}) <
\prod_{i=1}^{n-1}\bigg( 1-\frac{1}{n-i+\frac{i-1}{n-1}}\bigg) =
\frac{1}{n-1}\prod_{i=1}^{n-1}\bigg[1+\frac{1}{(n-2)(n-i)}\bigg]^{-1}.
\end{equation}

The measured failure rate for the SIS data in Table~\ref{tab:rnd} is $1-0.985472=0.014528$, not far from $1/64 = 0.015625$. A sample of $10^{4}$ runs of Algorithm~\ref{alg:sis} of $10^{6}$ de\-range\-ments each with $n=64$ gives an average failure rate of $0.01453(12)$ with a sample minimum of 0.014130 and maximum of 0.014991, where the digits within parentheses indicate the uncertainty at one standard deviation in the corresponding last digits of the datum. Figure~\ref{fig:eff} depicts Monte Carlo data for the failure probability (\ref{eq:fail}) against the upper bounds $1/n$ and (\ref{eq:approx}). Each data point was obtained as an average over $10^{4}$ runs of Algorithm~\ref{alg:sis} of $10^{6}$ de\-range\-ments each except for $n=512$, for which the runs are of $2 \times 10^{5}$ de\-range\-ments each.



\begin{figure}[ht]
\centering
\includegraphics[viewport=0 10 540 430,scale=0.32,clip]{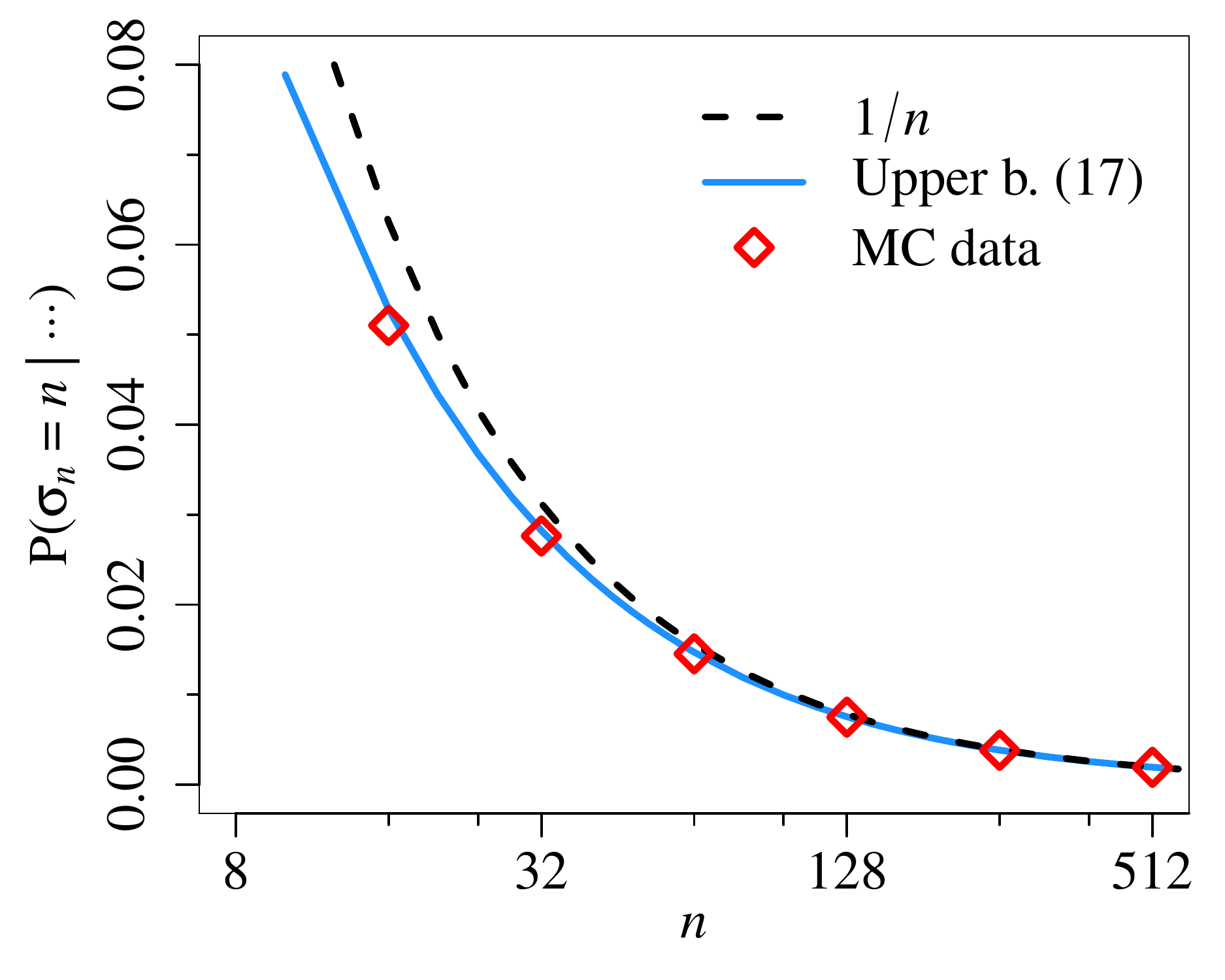}
\caption{\label{fig:eff}Measured failure rate for Algorithm~\ref{alg:sis} against $1/n$ and the upper bound (\ref{eq:approx}). Error bars in the data are much smaller than the symbols shown.}
\end{figure}


\subsection{\label{sec:uniform}Uniformity (or the lack thereof) of the SIS algorithm}

In the SIS approach, the ensuing sampling probabilities may deviate considerably from the uniform distribution. As it happens, Algorithm~\ref{alg:sis} does not generate each derangement in $D_{n}$ with uniform probability $1/d_{n}$. This can be seen by a simple pencil-simulation of the algorithm for some small $n$, say $n=3$. In this case $D_{3}=\{231,312\}$. To build $312$, we must first pick $3$ with probability $1/2$ then choose $1$ and $2$ sequentially, thus generating the derangement $312$ with probability $1/2$. To build $231$, otherwise, we must first pick $2$ with probability $1/2$, then $3$ with probability $1/2$ and then $1$ is forced, such that $231$ occurs with probability $1/4$. If we first pick $2$ and then $1$ the algorithm fails with probability $1/4$. We see that, by the rules of Algorithm~\ref{alg:sis}, $\PP(\sigma=312) \neq \PP(\sigma=231)$.

To verify whether the probability imbalance persists or smoothes out for larger $n$, we generate $100\,d_{n}$ derangements by Algorithm~\ref{alg:sis} for $n=8$ ($d_{8}=14833$) and $n=11$ ($d_{11}=14\,684\,570$) and bin the data. It is hard to run statistical tests involving all derangements for $n>11$ because either the sizes of the data files become humongous (hundreds of gigabytes if we insist in $100\,d_{n}$ samples) or the processing time becomes prohibitive (e.\,g., binning the derangements on the run involves searching). We found that Algorithm~\ref{alg:sis} indeed generated all derangements in $D_{8}$ and $D_{11}$ many times each in the runs. Figure~\ref{fig:freqDn}, however, definitely does not depict a distribution of occurrences peaked sharply about $100$ (the bins are of size $5$) which would represent uniform distribution. We cannot even argue that the distributions are becoming sharper with increasing $n$, since the standard deviation of the data are virtually the same in both cases: $\text{sd}_{8} \simeq 33.5$ versus $\text{sd}_{11} \simeq 32.4$.

Yet the data in Table~\ref{tab:rnd} clearly suggest that Algorithm~\ref{alg:sis} does sample $D_{n}$ according to the expected distribution of cycle lengths for $n=64$; the same behavior was also observed for a couple of other $n \geq 20$. A possible explanation is that $D_{n}$ is so large already for moderate values of $n$ (for instance, $d_{20} = 8.950 \times 10^{17}$), that any relatively ``small'' sample (which can actually be extremely large in absolute numbers) obtained by Algorithm~\ref{alg:sis} will most likely not contain repeated derangements. We verified this claim empirically: in five separate samples of $10^8$ derangements of $20$ labels each, not a single derangement occured twice either within a sample or between them. For practical purposes, then, Algorithm~\ref{alg:sis} samples $D_{n}$ ``uniformly.''

We could neither prove the uniformity nor the non-uniformity of Algorithm~\ref{alg:sis} rigorously. An attempt based on techniques borrowed from \cite{beichl,cdhl,kuznetsov,rasmussen} proved flawed. For one-sided restricted permutations of the type $\sigma_{i} \geq b_{i}$ for given $n \geq b_{n} \geq \cdots \geq b_{1} \geq 1$ (cf.~Remark~\ref{rmk:nlnn}), \citet{pd-rlg-sph} prove (Lemma~3.2) that a simple SIS algorithm samples all possible permutations uniformly and, moreover, that the algorithm never fails because of the particular form of the restrictions. A recent account on the SIS approach to sample one-sided restricted permutations is given by \citet{chung}. Their arguments do not seem to apply to de\-range\-ments, though.

\begin{figure}[t]
\centering
\includegraphics[viewport=10 0 480 490,scale=0.32,clip]{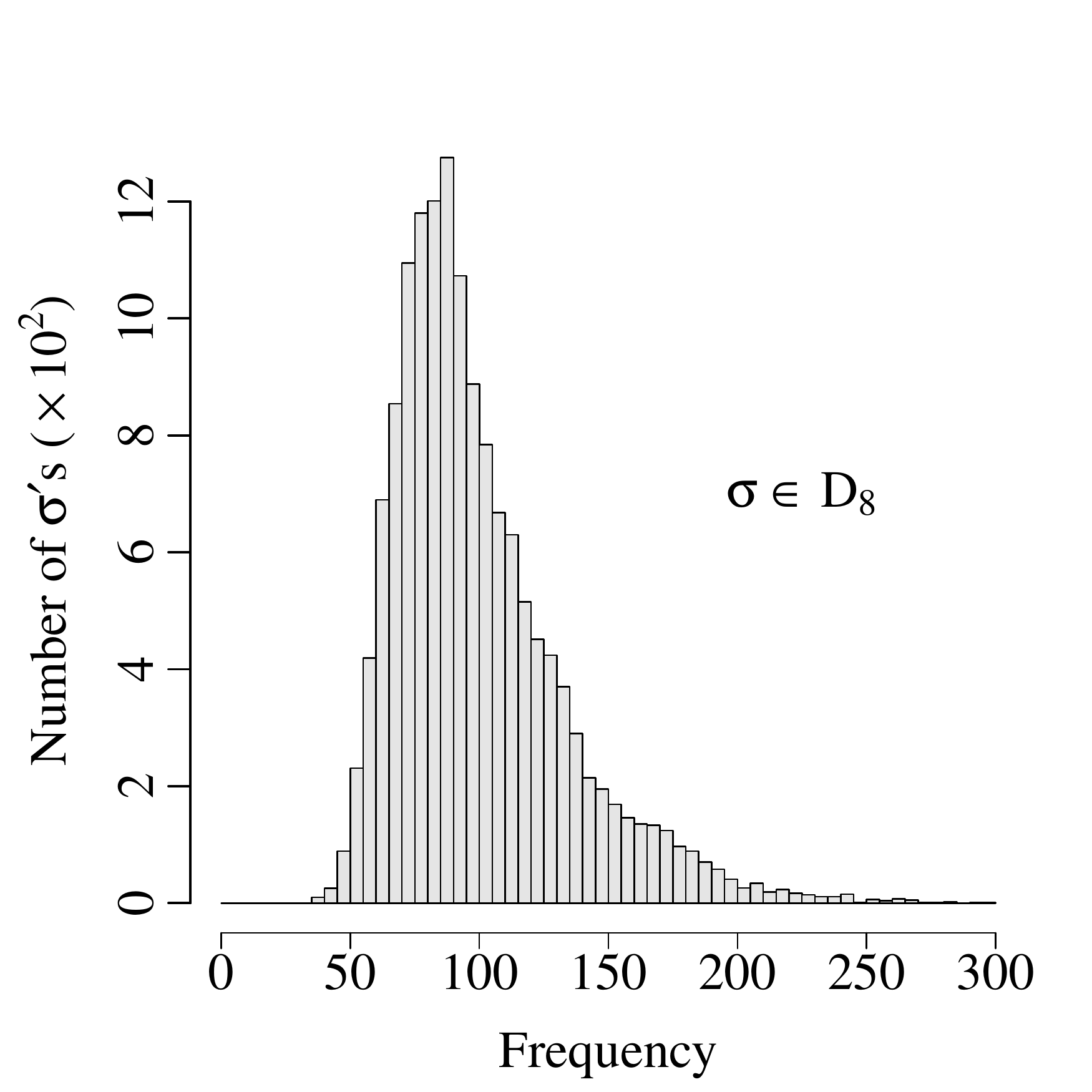}
\hspace{2em}
\includegraphics[viewport=10 0 480 490,scale=0.32,clip]{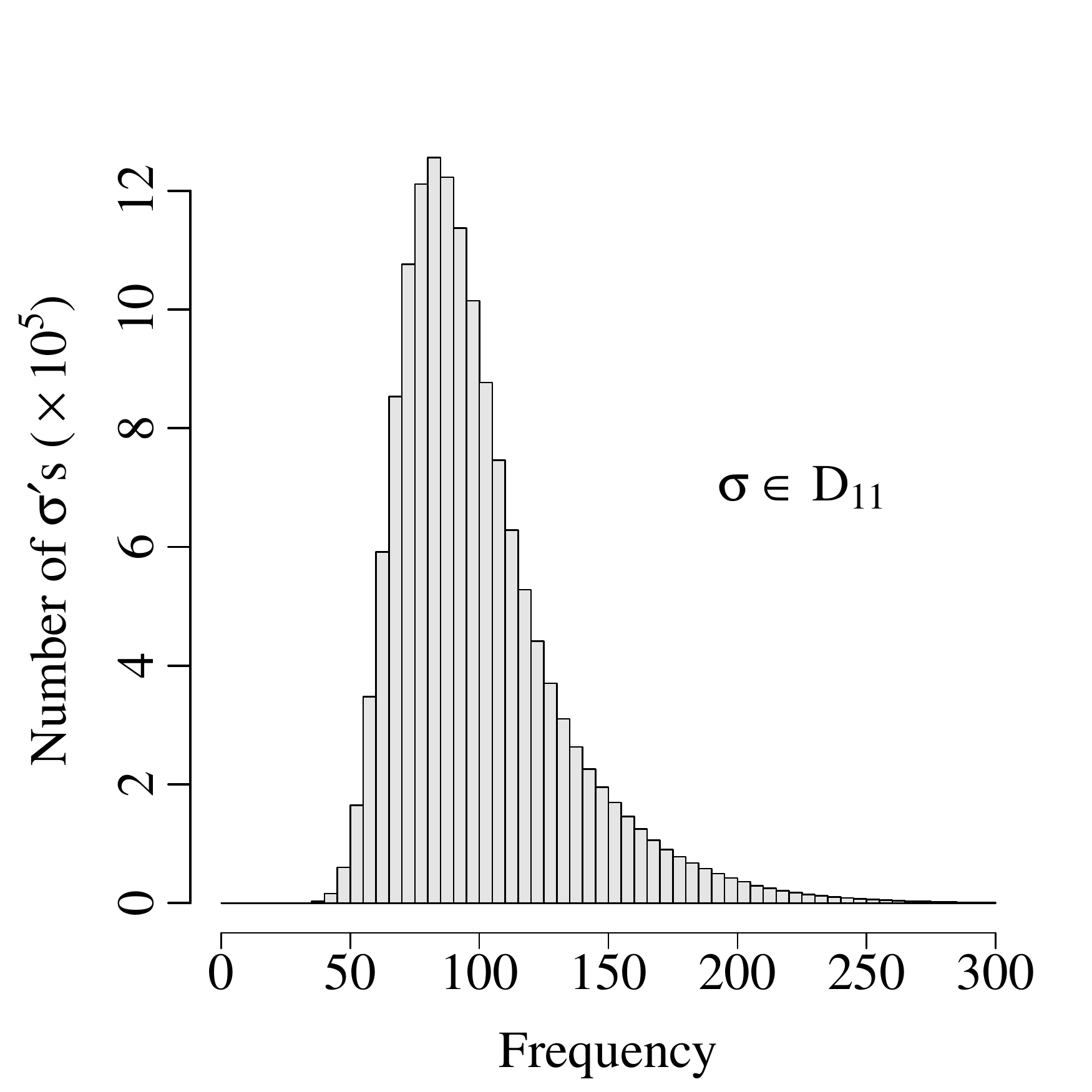}
\caption{\label{fig:freqDn}Number of derangements $\sigma \in D_{n}$ that occur in a sample of size $100d_{n}$ generated by Algorithm~\ref{alg:sis} for $n=8$ and $11$. The bins are of size $5$. Note the different vertical scales, as $d_{11}/d_{8} \simeq 990$. If Algorithm~\ref{alg:sis} sampled $D_{n}$ uniformly, we would expect a sharp peak of height $O(d_{n})$ located at frequency $100$.}
\end{figure}


\section{\label{sec:mix}Mixing time of the restricted transpositions shuffle}

To shed some light on the question of how many random restricted transpositions are necessary to generate random de\-range\-ments uniformly over ${D}_{n}$, we investigate the convergence of Algorithm~\ref{alg:switch} numerically. This can be done by monitoring the evolution of the empirical probabilities along the run of the algorithm towards the exact probabilities given by (\ref{eq:prob}).

Let $\nu$ be the measure that puts mass $\nu(k)={d}_{n}^{(k)}/{{d}_{n}}$ on the set ${D}_{n}^{(k)}$ and $\mu_{t}$ be the empirical measure
\begin{equation}
\label{eq:mu}
\mu_{t}(k) = \frac{1}{t}\sum_{s=1}^{t} \one\{\sigma_{s} \in {D}_{n}^{(k)}\},
\end{equation}
where $\sigma_{s}$ is the de\-range\-ment obtained after attempting $s$ restricted transpositions by Algorithm~\ref{alg:switch} on a given initial de\-range\-ment $\sigma_{0}$. The total variance distance between $\mu_{t}$ and $\nu$ is given by \citep{aldous,persi}
\begin{equation}
\label{eq:tvd}
d_{\mathrm{TV}}(t) = \|\mu_{t}-\nu\|_{\mathrm{TV}} = \frac{1}{2}\sum_{k=1}^{\lfloor n/2 \rfloor}\abs{\mu_{t}(k)-\nu(k)}.
\end{equation}
The right-hand side of (\ref{eq:tvd}) can be seen as the ``histogram distance'' between $\mu_{t}$ and $\nu$ in the $\ell_{1}$ norm. Clearly, $0 \leq d_{\mathrm{TV}}(t) \leq 1$. This distance allows us to define $t_{\mathrm{mix}}(\eps)$ as the time it takes for $\mu_{t}$ to fall within distance $\eps$ of $\nu$,
\begin{equation}
\label{eq:mix}
t_{\mathrm{mix}}(\eps)=\min\{t \geq 0 \colon d_{\mathrm{TV}}(t) < \eps\}.
\end{equation}
It is usual to define \textit{the} mixing time $t_{\mathrm{mix}}$ by setting $\eps=\frac{1}{4}$ or $\eps=\frac{1}{2}e^{-1} \simeq 0.184$, this last figure being reminiscent of the spectral analysis of Markov chains. We set $\eps=\frac{1}{2}e^{-1}$. This choice is motivated by the following pragmatic reasons:
\begin{itemize}
\item[(i)]We want the de\-range\-ments output by Algorithm~\ref{alg:switch} to be as uniformly distributed over $D_{n}$ as possible, so the smaller the $\eps$ the better the assessment of the algorithm and the choice of the constant $\mathit{mix}$;
\item[(ii)]Most of the probability mass is concentrated on a few cycle numbers (see Table~\ref{tab:rnd} and Remark~\ref{rmk:log} below), such that even relatively small differences between $\mu_{t}$ and $\nu$ are likely to induce noticeable biases in the output of Algorithm~\ref{alg:switch};
\item[(iii)]With $\eps=\frac{1}{4}$ we found that $t_{\mathrm{mix}}<n/2$, meaning that not even every possible de\-range\-ment had chance to be generated if the initial de\-range\-ment is cyclic (see Remark~\ref{rmk:split}).
\end{itemize}

\begin{remark}
\label{rmk:normal}
It is well known that the number of $k$-cycles of random $n$-permutations is Poisson distributed with mean $1/k$, such that as $n \nearrow \infty$ the CLT implies that the length of the cycles of random permutations follow a normal ditribution with mean $\log{n}$ and variance $\log{n}$; see, e.\,g., \cite{arratia} and the references there in. \citet{soria} proved that the same holds for permutations with no cycles of length less than a given $\ell>1$ using complex asymptotics of exponential generating functions; \citet{analco} and \citet{iran} provide the analysis for the particular case of de\-range\-ments. Figure~\ref{fig:pdf} displays the exact distribution of $k$-cycles for de\-range\-ments with $n=2^{15}=32768$ together with the normal density $N(\log{n},\sqrt{\mkern1mu\log{n}})$. For $n=32768$ we obtain from equations (\ref{eq:prob})--(\ref{eq:dnk}) that $\langle k \rangle = 9.967\cdots$ and $\sqrt{\langle k^{2} \rangle-\langle k \rangle^{2}} = 2.872\cdots$, while $\log{n} = 10.397\cdots$ and $\sqrt{\mkern1mu\log{n}} = 3.224\cdots$. The distribution of cycle lengths in Figure~\ref{fig:pdf} indeed looks close to a normal $N(\log{n},\sqrt{\mkern1mu\log{n}})$, albeit slightly skewed. We did not go to greater $n$ because Stirling numbers of the first kind are notoriously hard to compute even by computer algebra systems running on modern workstations. Recently, the cycle structure of certain types of restricted permutations (with $\sigma_{i} \geq i-1$) was also shown to be asymptotically normal \citep{ozel}.

\begin{figure}[t]
\centering
\includegraphics[viewport=10 0 480 490,scale=0.32,clip]{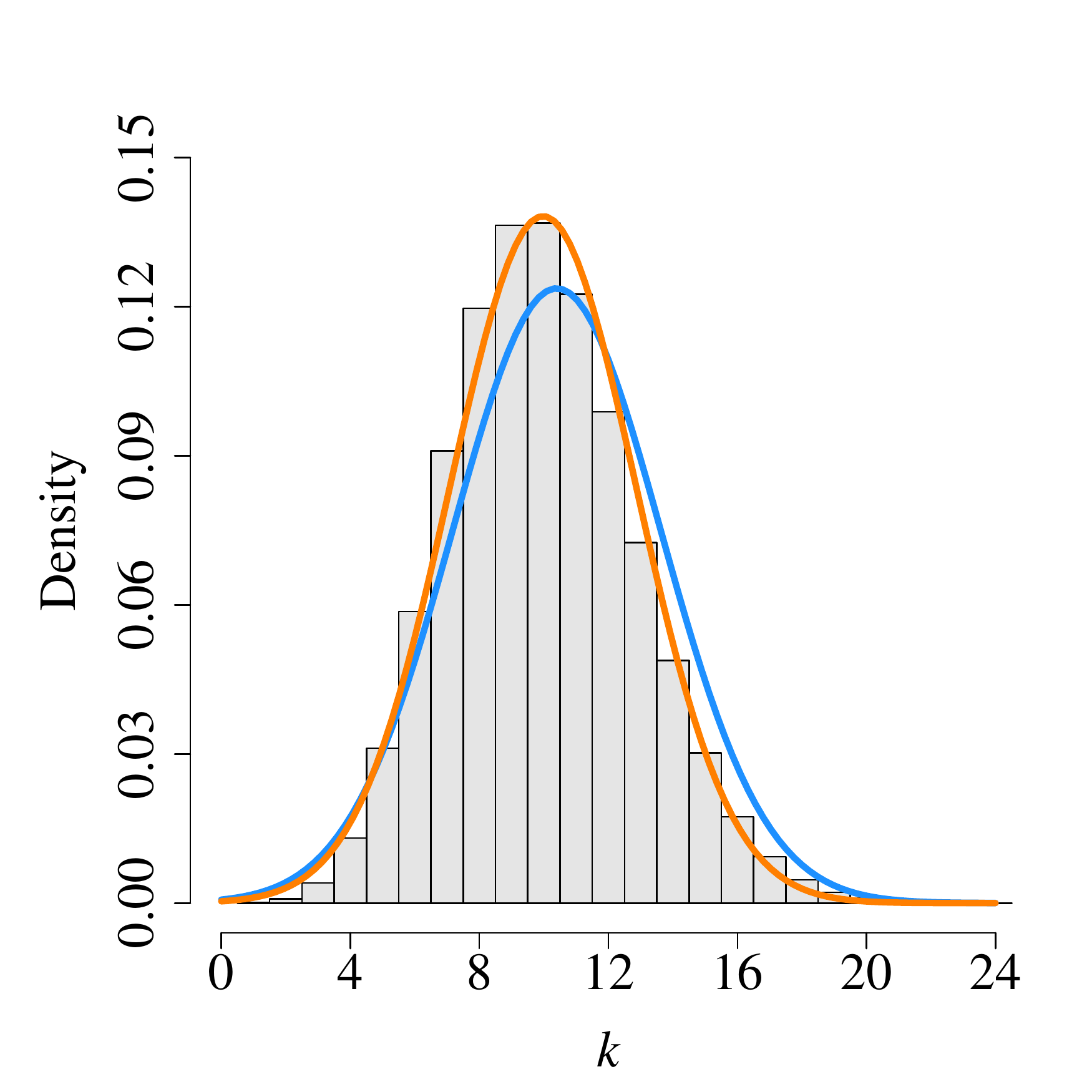}
\caption{\label{fig:pdf}Distribution of cycle lengths of random $n$-de\-range\-ments for $n=2^{15}=32768$ together with the normal densities $N(\log{n},\sqrt{\mkern1mu\log{n}})$ (shorter, in blue) and $N(m,s)$ (taller, in orange) with $m = \langle k \rangle = 9.967\cdots$ and $s = \sqrt{\langle k^{2} \rangle-\langle k \rangle^{2}} = 2.872\cdots$.}
\end{figure}
\label{rmk:log}
\end{remark}

Starting with a cyclic de\-range\-ment, i.\,e., with $\mu_{0}(1)=1$ and all other $\mu_{0}(k)=0$, we run Algorithm~\ref{alg:switch} and collect statistics on $d_{\mathrm{TV}}(t)$. Figure~\ref{fig:dtv} displays the average $\langle d_{\mathrm{TV}}(t) \rangle$ over $10^{6}$ runs for $n=128$. The behavior of $\langle d_{\mathrm{TV}}(t) \rangle$ does not show sign of the cutoff phenomenon---a sharp transition from unmixed state ($d_{\mathrm{TV}}(t_{\mathrm{mix}}-\delta) \approx 1$) to mixed state ($d_{\mathrm{TV}}(t_{\mathrm{mix}}+\delta) \approx 0$) over a small window of time $\delta \ll t_{\mathrm{mix}}$. Table~\ref{tab:mix} lists the average $\langle t_{\mathrm{mix}} \rangle$ obtained over $10^{6}$ samples for larger de\-range\-ments at $\eps = \frac{1}{2}e^{-1}$. An adjustment of the data to the form
\begin{equation}
\label{eq:adj}
t_{\mathrm{mix}} =c{\mkern1mu}n^{a}\log{n^{2}}
\end{equation}
furnishes $a=0.527(2)$ and $c=0.90(1)$. Our data thus suggest that $t_{\mathrm{mix}} \sim O(n^{a}\log{n}^{2})$ with $a \simeq \frac{1}{2}$, roughly an $O(\sqrt{n})$ lower than the upper bound given by \cite{aaron}. It is tempting to conjecture that $a=\frac{1}{2}$ (and, perhaps, that $c=1$) exactly, cf.\ last two lines of Table~\ref{tab:mix}, although our data do not support the case unequivocally.

\begin{figure}[t]
\centering
\includegraphics[viewport=10 10 540 430,scale=0.32,clip]{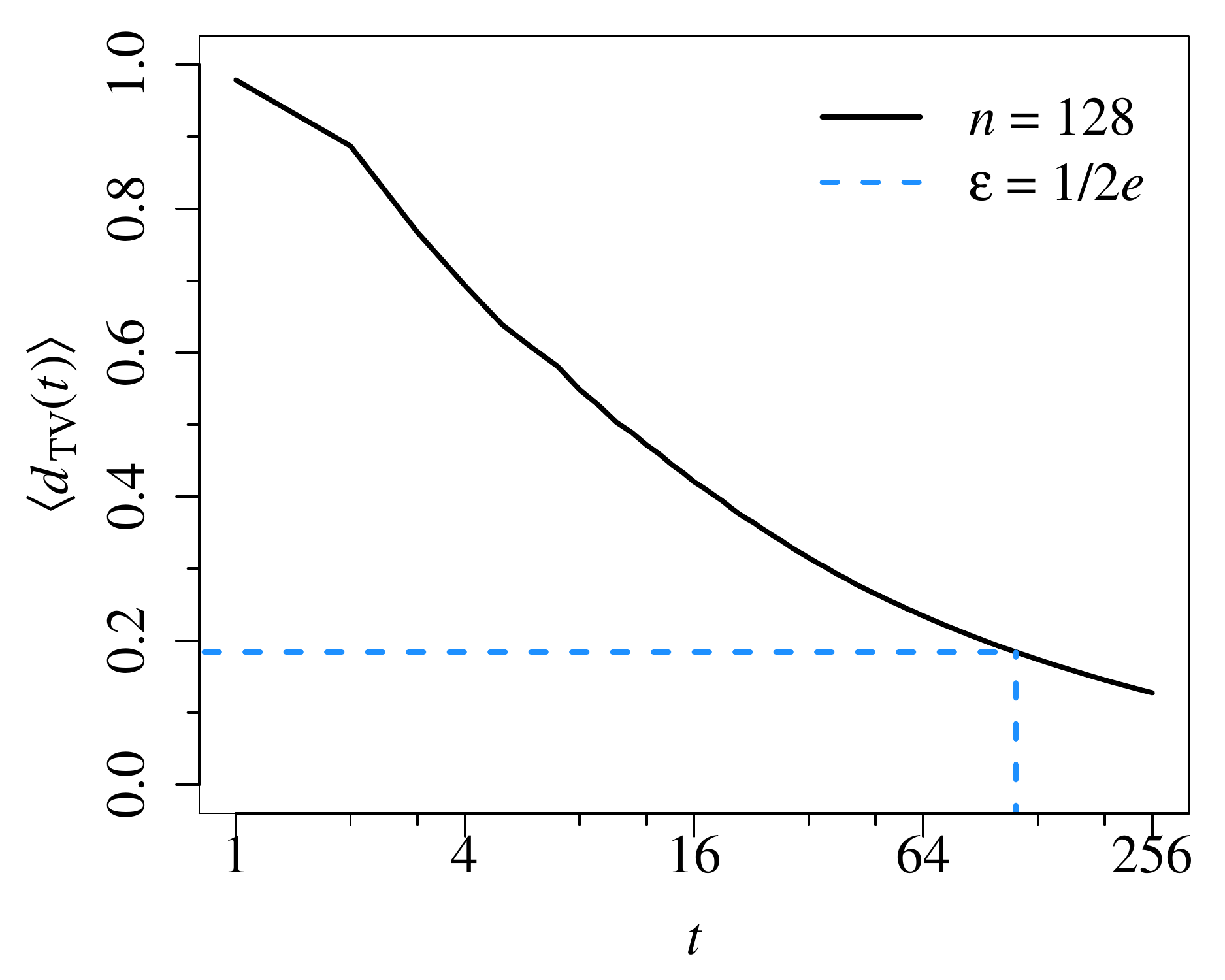}
\caption{\label{fig:dtv}Total variance distance $\langle d_{\mathrm{TV}}(t) \rangle$ (averaged over $10^{6}$ runs) between the empirical measure $\mu_{t}$ (with $\mu_{0}(1)=1$) and the stationary measure $\nu$ of the process defined by Algorithm~\ref{alg:switch} for $n=128$. The dashed line indicates the level $\eps=\frac{1}{2}e^{-1}$.}
\end{figure}

\begin{table}[t]
\caption{\label{tab:mix}Mixing time $t_{\mathrm{mix}}$ evaluated at $\eps=\frac{1}{2}e^{-1}$ obtained from an average trajectory $\langle d_{\mathrm{TV}}(t) \rangle$ over $10^{6}$ samples; see Figure~\ref{fig:dtv}. The second line displays the best guess to $n^{a}\log{n^{b}}$ involving only integer and semi-integer exponents. The last line displays the adjusted $a$ supposing a dependence like in (\ref{eq:adj}) with $c=1$.}
\centering
\begin{tabular}{ccccccccc}
\hline \\[-6pt]
$n$ & $64$ & $128$ & $192$ & $256$ & $320$ & $384$ & $448$ & $512$ \\[1pt]
\hline \\[-6pt]
$\langle t_{\mathrm{mix}} \rangle$ & $67$ & $112$ & $150$ & $184$ & $216$ & $245$ & $274$ & $301$ \\[1pt]
$\sqrt{n}\log{n^{2}}$ & $67$ & $110$ & $146$ & $177$ & $206$ & $233$ & $258$ & $282$ \\[1pt]
$a$ in $n^{a}\log{n^{2}}$ & $0.502$ & $ 0.504$ & $0.505$ & $0.507$ & $0.508$ & $0.508$ & $0.510$ & $0.510$ \\[1pt]
\hline
\end{tabular}
\end{table}


\section{\label{sec:summary}Summary and conclusions}

While simple rejection-sampling generates random de\-range\-ments with an acceptance rate of $\sim e^{-1}$ $\simeq 0.368$, thus being $O(e{\cdot}n)$ (plus the cost of verifying if the permutation generated is a de\-range\-ment, which does not impact the complexity of the algorithm but impacts its runtime), Sattolo's $O(n)$ algorithm only generates cyclic de\-range\-ments, and Mar\-t\'{\i}nez-Panholzer-Prodinger algorithm, with guaranteed uniformity, is $2n+O(\log^{2}n)$, we described two procedures, Algorithms~\ref{alg:switch} and \ref{alg:sis}, that are competitive for the efficient generation of random de\-range\-ments. In Section~\ref{sec:match} we discussed how Algorithm~\ref{alg:switch} can also be used, with $n$ even, to generate only random fixed-point-free involutions. Since fixed-point-free involutions of even $n$ labels can be viewed as perfect matchings on the complete graph, Algorithm~\ref{alg:switch} can become handy in a multitude of situations.

We found, numerically, that $O(n^{a}\log{n}^{2})$ random restricted transpositions with $a \simeq \frac{1}{2}$ suffice to spread an initial $n$-de\-range\-ment over ${D}_{n}$ measured by the distribution of cycle lengths. The fact that $2n > cn^{a}\log{n^{2}}$ for all $n \geq 1$ as long as $a \leq 0.63$ and $c \leq 1$ explains the good statistics displayed by Algorithm~\ref{alg:switch} with $\mathit{mix}=2n$, see Table~\ref{tab:rnd}. Currently, there are few analytical results on the mixing time of the random restricted transposition walk implemented by Algorithm~\ref{alg:switch}; the upper bound $O({n}\log{n^{2}})$ obtained by \citet{aaron} is roughly $O(\sqrt{n})$ above our numerical estimations. \citet{phylos,matchings} obtain a sharp $O({n}\log{n})$ estimate for the mixing time of a ``switch Markov chain'' for perfect matchings. Their chain builds perfect matchings as unordered sets $\{i,j\}$, not as ordered pairs $(i,j)$, as we do. Their numbers, however, are clearly equal because as a $2$-cycle $(ij) \equiv (ji)$. It would be interesting to run Algorithm~\ref{alg:switch} in the ``perfect matchings mode'' to check whether its mixing time display a different behavior.

Algorithm~\ref{alg:switch} employs $2\,\mathit{mix}$ pseudorandom numbers and Algorithm~\ref{alg:sis} employs $O(n)$ pseudorandom numbers to generate an $n$-de\-range\-ment distributed over ${D}_{n}$ with the expected distributions of cycle lengths. In this way, even if we set $\mathit{mix} = c\sqrt{n}\,\log{n}^{2}$ with some $1 < c \sim O(1)$, both algorithms perform better than currently known methods, with comparable runtime performances between them. As we argued in Section~\ref{sec:uniform}, for relatively small samples, which can actually be very large in absolute numbers (several billion derangements, for instance) since $D_{n}$ is such a huge set already for moderate $n$, in practice Algorithm~\ref{alg:sis} samples derangements ``uniformly.''


\section*{Acknowledgments}

The author thanks Aaron Smith (U. Ottawa) for useful correspondence and suggestions improving a previous version of the manuscript, the Laboratoire de Physique Th\'{e}orique et Mod\`{e}les Statistiques -- LPTMS (CNRS UMR 8486) for kind hospitality during a sabbatical leave in France where part of this work was done, and FAPESP (Brazil) for partial support through grant no.~2017/22166-9.


\end{document}